\documentclass[conference]{IEEEtran}
\usepackage{epsfig,latexsym,amsmath,amssymb,setspace,cite,color,graphicx,algorithm,algorithmic,cases}
\usepackage{amssymb,amsmath,amsfonts,bm,epsfig,graphicx,latexsym}
\usepackage{rotating,setspace,latexsym,epsf,color}
\usepackage{cite,authblk,epstopdf,footnote}
\usepackage{float}
\usepackage[inline]{enumitem}
\usepackage{tabularx}
\usepackage{array}
\usepackage{bbm}
\usepackage{amsthm}
\usepackage{thmtools}
\usepackage[top=0.71in, left=0.67in, bottom=1in, right=0.67in, columnsep=0.23 in]{geometry}
\usepackage{algorithm,algorithmic}

\declaretheorem[style=plain,name=Definition,qed=$\blacksquare$]{Definition}
\declaretheorem[style=plain,name=Lemma]{lemma}
\declaretheorem[style=plain,name=Remark,qed=$\blacksquare$]{remark}
\declaretheorem[style=plain,name=Example,qed=$\blacksquare$]{example}
\declaretheorem[style=plain,name=Proposition,qed=$\blacksquare$]{proposition}
\declaretheorem[style=plain,name=Corollary,qed=$\blacksquare$]{corollary}
\declaretheorem[style=plain,qed=$\blacksquare$]{theorem}
\allowdisplaybreaks
\def\mc{\ensuremath\mathcal}

\begin{document}


\title{On Coded Caching with Heterogeneous Distortion Requirements} 
\author{Abdelrahman M. Ibrahim}
\author{Ahmed A. Zewail}
\author{Aylin Yener}
\affil{\normalsize Wireless Communications and Networking Laboratory (WCAN)\\
School of Electrical Engineering and Computer Science\\
The Pennsylvania State University, University Park, PA 16802.\\
\em \{ami137,zewail\}@psu.edu \qquad yener@engr.psu.edu}

\maketitle

\begin{abstract}
This paper considers heterogeneous coded caching where the users have unequal distortion requirements. The server is connected to the users via an error-free multicast link and designs the users' cache sizes subject to a total memory budget. In particular, in the placement phase, the server jointly designs the users' cache sizes and the cache contents. To serve the users' requests, in the delivery phase, the server transmits signals that satisfy the users' distortion requirements. An optimization problem with the objective of minimizing the worst-case delivery load subject to the total cache memory budget and users' distortion requirements is formulated. The optimal solution for uncoded placement and linear delivery is characterized explicitly and is shown to exhibit a threshold policy with respect to the total cache memory budget. As a byproduct of the study, a caching scheme for systems with fixed cache sizes that outperforms the state-of-art is presented.  

\end{abstract}


\section{Introduction}
Wireless data traffic is increasing at an unprecedented rate due to the demand on video streaming services, which have accounted for $60$ percent of total mobile data traffic in $2016$ \cite{cisco}. Efficient utilization of network resources is essential in order to accommodate the growth in data traffic. Caching utilizes the cache memories in the network nodes to shift some of the data traffic to off-peak hours. Reference \cite{maddah2014fundamental} has proposed \textit{coded caching}, where the end-users' cache contents are designed in a way to enable subsequently serving the users using multicast transmissions, reducing the delivery load on the server.


The significant gain achieved by coded caching has motivated studying its fundamental limits in various setups \cite{maddah2014fundamental,maddah2016coding,wan2016optimality,
yu2016exact,yu2017characterizing_ISIT,karamchandani2016hierarchical,ji2015comb,zewail2017combination,
wan2017combination,wan2017caching,ji2016fundamental,zewail2016fundamental,ibrahim2018device,
ravindrakumar2016fundamental,sengupta2015fundamental,zewail2016coded,
wang2015fundamental,amiri2016decentralized,sengupta2016layered,ibrahim2017centralized,
ibrahim2018coded,bidokhti2017Gaussian,bidokhti2017benefits2,
bidokhti2017benefits,ibrahim2017optimization,hassanzadeh2015distortion,
yang2016coded}. For instance, references \cite{karamchandani2016hierarchical,ji2015comb,zewail2017combination,
wan2017combination,wan2017caching} have studied coded caching in multi-hop networks. Coded caching for device-to-device (D2D) networks has been studied in \cite{ji2016fundamental,zewail2016fundamental,ibrahim2018device}. Caching with security requirements have been studied in \cite{ravindrakumar2016fundamental,sengupta2015fundamental,zewail2016coded,zewail2017combination}.

Content delivery networks consists of heterogeneous end-devices that have varying storage, computational capabilities, and viewing preferences. In turn, the effect of heterogeneity in cache sizes on the delivery load has been studied in \cite{wang2015fundamental,amiri2016decentralized,sengupta2016layered,
ibrahim2017centralized,ibrahim2018coded,ibrahim2018device}. Additionally, optimizing the users' cache sizes subject to a network-wide total memory budget has been considered in \cite{ibrahim2017optimization,bidokhti2017Gaussian,
bidokhti2017benefits2,bidokhti2017benefits}.


The heterogeneity in users' preferences for content motivates developing coded caching schemes with different quality-of-service requirements per user. In this realm, coded caching schemes with heterogeneous distortion requirements have been studied in references \cite{hassanzadeh2015distortion,yang2016coded,timo2017rate}. In particular, reference \cite{hassanzadeh2015distortion} has considered a centralized system where files are modeled as independent and identically distributed (i.i.d.) samples from a Gaussian source. Each file is represented by a number of layers equal to the number of users in the system, and accessing the first $k$ layers guarantees that the $k$th user's distortion requirement is satisfied. In this setup, reference \cite{hassanzadeh2015distortion} has minimized the squared error distortion for given delivery load, cache sizes, and popularity distribution. Reference \cite{yang2016coded} has studied the problem of minimizing the delivery load in a centralized caching system with heterogeneous distortion requirements and heterogeneous cache sizes at the users. In particular, reference \cite{yang2016coded} has considered a separation approach where the memory allocation over the layers and the caching scheme are optimized separately. 


In this work, like reference \cite{yang2016coded}, we study the problem of minimizing the delivery load given heterogeneous distortion requirements at the users. Different from references \cite{hassanzadeh2015distortion} and \cite{yang2016coded}, we consider that the server not only designs the users' cache contents, but also optimizes their cache sizes subject to a total cache memory budget. In contrast to \cite{yang2016coded}, we jointly design the users cache sizes, the memory allocation over the layers, and the caching scheme in order to minimize the delivery load that achieves certain distortion requirements at the users. Under uncoded placement and linear delivery, we show that the joint optimization problem reduces to a memory allocation problem over the layers which can be solved analytically. In particular, the memory allocation over the layers is obtained using a threshold policy, which depends on the available total cache memory budget and the target distortion requirements at the users. We extend the cut-set bound in \cite{yang2016coded} to systems with total cache memory budget and compare it with the delivery load achieved by the proposed scheme. We observe that the cut-set bound is achievable for large total cache memory budgets.

Although our primary goal in this study is to demonstrate the merit of optimally allocating the cache sizes at different users, we note that the new caching scheme we propose improves on the caching schemes in \cite{yang2016coded} for systems with fixed cache sizes, since we jointly optimize the caching scheme and the memory allocation over the layers. More specifically, the flexibility in our scheme allows us to exploit the multicast opportunities over all layers. Our numerical results confirm the gain attained by our proposed scheme over the two caching schemes presented in \cite{yang2016coded}. Finally, we present a numerical example that shows the suboptimality of exploiting only the intra-layer multicast opportunities without taking into account the inter-layer multicast opportunities.

\textit{Notation:} Vectors are represented by boldface letters, $ \oplus$ refers to bitwise XOR operation, $|W|$ denotes size of $W$, $\mc A \setminus \mc B $ denotes the set of elements in $\mc A$ and not in $\mc B $, $[K] \triangleq \{1,\dots,K\}$, $ \phi$ denotes the empty set, $\subsetneq_{\phi} [K]$ denotes non-empty subsets of $[K]$.

\begin{figure}[t]
	\includegraphics[scale=1.2]{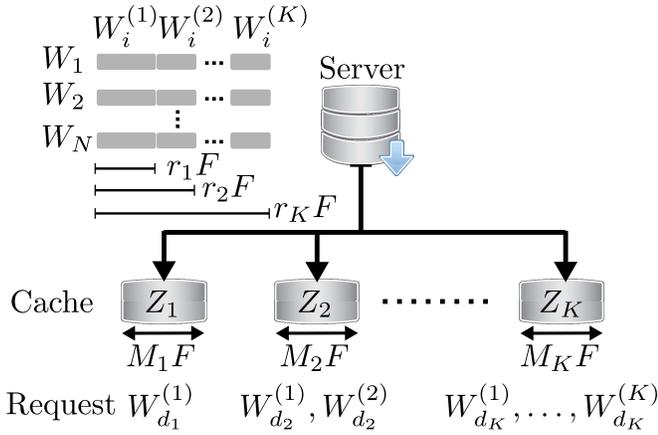}
	\centering
	\caption{Caching system with heterogeneous distortion requirements $D_1 \geq D_2 \geq \dots \geq D_K $.}\label{fig:sys_model}
	\vspace{-.2 in}
\end{figure}


\section{System Model}\label{sec_sysmod}
\vspace{+.09 in}
We consider a server connected to $K$ users via an error-free multicast network, see Fig. \ref{fig:sys_model}. The server has a library of $N$ files, $ W_{1}, \dots, W_{N} $, each file $W_j$ consists of $F$ i.i.d. samples $ [S_{j,1},\dots,S_{j,F}]$ from a uniform discrete source. That is, $S_{j,i}$ is uniformly distributed over $\mathbb{F}_q$, $Pr(S_{j,i}=s)=1/q$ for $s \in \mathbb{F}_q$. We use Hamming distortion as our distortion measure, i.e., we have 
\begin{align}
d(s,\hat{s})= \begin{cases}0, \text{ if } s=\hat{s},\\
1, \text{ if } s\neq\hat{s}.
\end{cases} 
\end{align}
The rate-distortion function for a uniform discrete source with Hamming distortion is given as \cite{cover2006elements} 
\begin{align}
\rho(D)=\log(q)-H(D)-D\log(q-1),
\end{align}
where $ 0 \leq D \leq 1-1/q$. Since a uniform discrete source with Hamming distortion is successively refinable \cite{equitz1991successive}, each file can be represented using a scalable layered description that achieves the target rate $r_l \triangleq \rho(D_l)$ at layer $l$ for distortion requirements $D_1 \geq D_2 \geq \dots \geq D_K $. In particular, layer $1$ of each file is a coarse description with size $r_1 F$ bits, while layer $l \in \{2,\dots,K\}$ is a refinement with size $(r_l-r_{l-1})F $ bits. In turn, we represent file $W_{j}$ by $W_{j}^{(1)},\dots,W_{j}^{(K)} $, where $W_{j}^{(l)} $ denotes layer $l$ of file $j$. Additionally, we have that $D_k$ is the distortion requirement at user $k$, i.e., user $k$ needs layers $\{1,\dots,k\}$ in order to decode the requested file. Consequently, we define the target rate vector $\bm r=[r_1,..,r_K]$.

The sizes of the users' cache memories are determined by the server. In particular, the server allocates $M_k F$ bits to user $k$ such that $\sum_{k=1}^{K} M_k F \leq  m_{\text{tot}} N F$ bits, where $m_{\text{tot}} $ is the cache memory budget normalized by library size $NF$. We also define $m_k=M_k/N$, to denote the memory size at user $k$ normalized by the library size $N F$. We consider the regime where the number of files is greater than or equal the number of users, i.e., $N \geq K $, and $M_k \in [0,N r_k], \ \forall k \in [K]$ which implies $m_k \in [0,r_k], \ \forall k \in [K]$. We denote the memory size vector by $\bm M =[M_1,\dots,M_K]$ and its normalized version by $\bm m =[m_1,\dots,m_K]$. 

The system has two operational phases: placement and delivery  \cite{ibrahim2017optimization}. In the placement phase, the server assigns the users' cache sizes and the contents of the users' cache memories subject to the distortion requirements $\bm D=[D_1,\dots,D_K]$ and the cache memory budget $m_{\text{tot}}$, without the knowledge of the users' demands. In particular, the server places a subset of the library, $Z_k$, at the cache memory of user $k$, such that $|Z_k| \leq M_k F$ bits. In the delivery phase, user $k$ requests the file $W_{d_k}$ which must be recovered with an average distortion less than or equal to $D_k$. Equivalently, user $k$ requests layers $[k]$ from file $W_{d_k}$. The requested files are represented by the demand vector $\bm d=[d_1, \dots, d_K]$. We assume that $\bm d$ consists of identical and independent uniform random variables over the files \cite{maddah2014fundamental}. In order to guarantee the desired distortion requirements at the users, the server needs to deliver the bits that are not cached by the users in each requested layer. In particular, the server sends the signals $X_{\mc T, \bm d}$ to the users in the sets $ \mc T \subsetneq_{\phi} [K]$. User $k$ should be able to decode $W^{(1)}_{d_k},\dots,W^{(k)}_{d_k}$ by utilizing the cached contents $Z_k$ and the transmitted signals $X_{\mc T, \bm d}, \mc T \subsetneq_{\phi} [K]$. Formally, we have the following definitions.
\begin{Definition} For a given file size $F$, a caching scheme is defined by the collection of cache placement, encoding, and decoding functions $(\varphi_k(.),\psi_{\mc T, \bm d}(.),\mu_{\bm d, k}(.))$. A cache placement function 
\begin{align}
\varphi_k: \mathbb{F}_q^F \times \dots \times \mathbb{F}_q^F \rightarrow [2^{M_kF}],
\end{align}
maps the $N$ files to the cache of user $k$, i.e., $ Z_k = \varphi_k(W_1, W_2,..,W_N) $. Given a demand vector $\bm d$, an encoding function 
\begin{align}
\psi_{\mc T, \bm d}: \mathbb{F}_q^F \times \dots \times \mathbb{F}_q^F \rightarrow [2^{v_{\mc T}F}],
\end{align}
maps the requested files to a signal with length $v_{\mc T} F$ bits, which is sent to the users in $\mc T$, i.e., $X_{\mc T,\bm d}= \psi_{\mc T, \bm d}( W_{d_{1}},..,W_{d_{K}})$. Finally, a decoding function 
\begin{align}
\mu_{\bm d, k}: [2^{M_kF}] \times [2^{RF}] \rightarrow \mathbb{F}_q^F, 
\end{align}
with  $R \triangleq \sum_{\mc T \subsetneq_{\phi} [K]} v_{\mc T}$, maps $Z_k$ and $X_{\mc T, \bm d},$ $\mc T \subsetneq_{\phi} [K]$ to $\hat W_{d_k}$, i.e., $\hat W_{d_k} = \mu_{\bm d, k}\left(X_{\{1\}, \bm d},X_{\{2\}, \bm d},\dots,X_{[K], \bm d}, Z_k \right)$. 
\end{Definition}

\begin{Definition}For given normalized cache sizes $\bm m$, distortion requirements $D_1 \geq D_2 \geq \dots \geq D_K $, and $r_l =\rho(D_l) $, the delivery load $R(\bm m, \bm r)$ is achievable if there exists a sequence of caching schemes such that $\lim\limits_{F \rightarrow \infty} \dfrac{1}{F} \sum\limits_{i=1}^{F} d(S_{d_k,i},\hat S_{d_k,i}) \leq D_k$, $\forall k \in [K]$, $\forall d_k \in [N]$. Furthermore, the infimum over all achievable delivery loads is denoted by $R^{*}(\bm m, \bm r)$. 
\end{Definition}
In this work, we consider the class of cache placement schemes, $\mathfrak{A}$, in which user $k$ cache uncoded pieces of layers $[k]$ of the files, and the class of delivery schemes, $\mathfrak{D}$, where the multicast signals are formed using linear codes. Note that the uniform demands assumption implies that each user should cache the same number of bits from all files, i.e., user $k$ dedicates $m_k F$ bits to each file.
\begin{Definition} For given normalized cache sizes $\bm m$ and target rates $\bm r$, 
the worst-case delivery load under an uncoded placement scheme in $ \mathfrak A$, and a linear delivery policy in $\mathfrak{D}$, is defined as 
\begin{equation}
R_{\mathfrak A, \mathfrak D}(\bm m, \bm r) \triangleq \max_{\bm d \in [N]^{K}}  R_{\bm d,\mathfrak A, \mathfrak D}  =  \sum_{\mc T \subsetneq_{\phi} [K]} v_{ \mc T}.
\end{equation}
Furthermore, by taking the infimum over $ \mathfrak A$ and all possible delivery policies, we get $R^*_{\mathfrak A}(\bm m, \bm r)$.
\end{Definition}
The trade-off between the delivery load and the total cache memory budget is defined as follows.
\begin{Definition} For given normalized cache memory budget $m_{\text{tot}}$ and target rate vector $\bm r$, the infimum over all achievable worst-case delivery loads is given as
\begin{align}
R^{*}(m_{\text{tot}}, \bm r)= \inf_{\bm m \in \mc M(m_{\text{tot}}, \bm r)} R^{*}(\bm m, \bm r),
\end{align}
where $\mc M(m_{\text{tot}}, \bm r)=\Big\{\bm m \big\vert \ \! 0 \! \leq \! m_k \! \leq \! r_k, \ \! \sum\limits_{k=1}^{K} m_k = m_{\text{tot}} \Big\} $. Furthermore, we have
\begin{align}
R^{*}_{\mathfrak A}(m_{\text{tot}}, \bm r)= \inf_{\bm m \in \mc M(m_{\text{tot}}, \bm r)} R^{*}_{\mathfrak A}(\bm m, \bm r).
\end{align}
\end{Definition}
\section{A Novel Caching Scheme}\label{sec_cach}
Given the layered description of the files explained in Section \ref{sec_sysmod}, the problem of designing the users' cache contents, under the uncoded placement assumption, can be decomposed into $K$ placement problems, each of which corresponds to one of the layers \cite{yang2016coded}. In particular, the cache memory at user $k$ is partitioned over layers $[k]$ and $m_{k}^{(l)} $ denotes the normalized cache memory dedicated to layer $l$ at user $k$. In turn, the placement problem of layer $l$ is equivalent to the placement problem with $K-l+1$ users and unequal cache memories $\bm m^{(l)}=[m_l^{(l)},\dots,m_K^{(l)}]$ addressed in \cite{ibrahim2017optimization,ibrahim2018coded}. By contrast, in the delivery phase, decoupling over the layers is suboptimal. That is, we need to jointly design the multicast signals over all layers, in order to utilize all multicast opportunities. Next, we explain the cache placement and delivery schemes in detail.

%
%
\subsection{Placement Phase}
Under uncoded placement, the placement phase is decomposed into $K$ layers. In particular,  layer $l$ of each file is partitioned over users $\{l,\dots,K\}$. That is, $W^{(l)}_{j}$ is divided into subfiles, $\tilde W^{(l)}_{j,\mc S}, \mc S \subset \{l,\dots,K\}$, which are labeled by the set of users $\mc S $ exclusively storing them. We assume that $|\tilde W^{(l)}_{j,\mc S}|= a^{(l)}_{\mc S} F$ bits $\forall j \in [N]$, i.e., we have symmetric partitioning over the files and the \textit{allocation variable} $a^{(l)}_{\mc S} \in [0,r_l\!-\!r_{l-1}]$ defines the size of $\tilde W^{(l)}_{j,\mc S}$. In turn, the set of feasible placement schemes for layer $l$ is defined by 
\begin{align}\label{eqn_feas_alloc}
  \mathfrak A^{(l)}(\bm m^{(l)},\bm r) &=   \bigg\lbrace  \bm a^{(l)} \Big\vert   \sum\limits_{\mc S \subset \{l,\dots,K\} } \! \! \! \!   a^{(l)}_{\mc S }=r_l-r_{l-1}, \nonumber \\  &\sum\limits_{\mc S \subset \{l,\dots,K\}: \ \! k \in \mc S } \! \! \! \! \! \! \! \! \! \! \! a^{(l)}_{\mc S } \leq m_k^{(l)}, \forall k \in \{l,\dots,K\} \bigg\rbrace,\!\!
\end{align}
where $\bm a^{(l)} $ is the vector representation of $\{ a^{(l)}_{\mc S} \}_{ \mc S} $. Note that the first constraint follows from the fact that the $l$th layer of each file is partitioned over the sets $\mc S \subset \{l,\dots,K\}$, while the second constraint represents the cache size constraint at layer $l$ for user $k$. Therefore, the cache content placed at user $k$ is given by
\begin{align}
Z_k = \bigcup\limits_{l \in [k]} \ \bigcup\limits_{j \in [N]} \ \bigcup\limits_{\mc S \subset \{l,\dots,K\}: \ \! k \in \mc S } \tilde W^{(l)}_{j,\mc S}.
\end{align}

\subsection{Delivery Phase}
In order to deliver the missing subfiles, the server sends the sequence of unicast/multicast signals $X_{\mc T, \bm d},$ $ \mc T \subset [K]$. In particular, the multicast signal to the users in $\mc T$ is defined by  
\begin{align}\label{eqn_mult_def1}
X_{\mc T, \bm d}= \oplus_{j \in \mc T} \ \bigcup_{l=1}^{L} W^{(l),\mc T}_{d_j}, 
\end{align}
where $L \triangleq \min_{i \in \mc T} i$ and $W^{(l),\mc T}_{d_j}$ denotes a subset of $W^{(l)}_{d_j} $ which is delivered to user $j$ via $X_{\mc T, \bm d}$. $W^{(l),\mc T}_{d_j}$ is constructed using the side-information available at the users in $\mc T \setminus \{j\} $. That is, if $W_{d_j,\mc S}^{(l),\mc T} $ denotes the subset of $W^{(l),\mc T}_{d_j}$ cached by the users in $\mc S$, then we have
\begin{align}\label{eqn_mult_def2}
W^{(l),\mc T}_{d_j} = \bigcup_{\mc S \subset \{l,\dots,K\}\setminus\{j\}: \ \! \mc T \setminus \{j\} \subset \mc S} W_{d_j,\mc S}^{(l),\mc T}.
\end{align}
Additionally, we denote $ |X_{\mc T, \bm d}|=\sum_{l=1}^{L} |W^{(l),\mc T}_{d_j}|= v_{\mc T} F$ bits $\forall j \in \mc T$ and $|W_{d_j,\mc S}^{(l),\mc T}|= u^{(l),\mc T}_{\mc S} F $ bits. That is, the transmission variable $v_{\mc T} \in [0,r_L]$ and the assignment variable $u^{(l),\mc T}_{\mc S} \in [0,a^{(l)}_{\mc S}]$ determine the structure of the signal $X_{\mc T, \bm d} $. Furthermore, the unicast signal $X_{\{k\}, \bm d} $ delivers the missing pieces of $ \bigcup_{l=1}^{k} W^{(l)}_{d_k}$ to user $k$, i.e., the pieces that had not been delivered by multicast signals and are not cached by user $k$. We assume $|X_{\{k\}, \bm d}|= \sum_{l=1}^{k} v_{\{k\}}^{(l)} F$ bits.

Next, we explain that all linear delivery schemes under uncoded placement can be described by the following linear constraints on the transmission and assignment variables. In particular, for given allocation $\{ \bm a^{(l)} \}_l$, the set of linear delivery schemes, $\mathfrak D(\bm a^{(1)},\dots,\bm a^{(K)},\bm r) $, is defined by
\begin{align}
&\! \! v_{\mc T} \! = \! \sum_{l=1}^{L} \sum_{\mc S \in \mc B^{(l),\mc T}_{j}} \! \!  \! \! u^{(l),\mc T}_{\mc S},  \forall \mc T \! \subset \! [K] \text{ s.t. } |\mc T| \geq 2, \forall j \! \in \! \mc T, \label{eqn_dlv_cnst1} \\
& \! \! \! \!  v_{\{k\}}^{(l)} +\sum\limits_{\mc T \subset \{l,\dots,K\} :  k \in \mc T, |\mc T|\geq 2} \sum_{\mc S \in \mc B^{(l),\mc T}_{k}} \! \!  \! \! u^{(l),\mc T}_{\mc S} + \! \! \! \! \! \! \! \! \sum_{\mc S \subset \{l,\dots,K\} :  k \in \mc S}  \! \! \! \! \! \! \! \! \! \!  a_{\mc S}^{(l)} \! \geq   \nonumber \\ & \qquad \qquad \qquad \quad  \! r_l \! - \! r_{l-1}, \forall l \in [K], \ \forall k \! \in \! \{l,\dots,K\}, \label{eqn_dlv_cnst2}\\
&\! \! \! \! \sum\limits_{\mc T \subset \{l,\dots,K\} : \ \! j \in \mc T, \mc T \cap \mc S \neq \phi} \! \! \! \! \! \! \! \! \! \! \! \! \! \! \! \! \! \!\!\!\! \! \! u^{(l),\mc T }_{\mc S} \leq a_{\mc S}^{(l)}, \forall l \! \in \! [K], \forall j \not\in \mc S, \forall \mc S \in \mc A^{(l)}, \label{eqn_dlv_cnst3}\\
&\! \! \! 0 \! \leq \! u_{\mc S}^{(l),\mc T} \! \! \leq \! a_{\mc S}^{(l)},\forall l \! \in \! [K], \forall \mc T \! \subsetneq_{\phi} \! \{l,\dots,K\}, \forall \mc S \! \in \! \mc B^{(l),\mc T} \! \!, \label{eqn_dlv_cnst4}
\end{align}
where 
\begin{align*}
\mc B^{(l),\mc T}_{j} &\triangleq  \big\lbrace \mc S \subset \{l,\dots,K\}\setminus\{j\}: \ \! \mc T \setminus \{j\} \subset \mc S \big\rbrace,\\
\mc A^{(l)} &\triangleq \big\lbrace \mc S \subset \{l,\dots,K\} : 2 \leq |\mc S| \leq K-l \big\rbrace,
\end{align*}
and $ \mc B^{(l),\mc T} \triangleq \bigcup_{j \in \mc T} \mc B^{(l),\mc T}_{j} $.


The structural constraints in (\ref{eqn_dlv_cnst1}) follows from the structure of the multicast signals in (\ref{eqn_mult_def1}) and (\ref{eqn_mult_def2}). The delivery completion constraints in (\ref{eqn_dlv_cnst2}) guarantee that the unicast and multicast signals complete the $l$th layer of the requested files, and the redundancy constraints in (\ref{eqn_dlv_cnst3}) prevent the transmission of redundant bits to the users. For example, for $K=3$, the structural constraints are defined as
\begin{subequations}
\begin{align}
v_{\{1,2\}} \! &= \! u^{(1),\{1,2\}}_{\{2\}} \! + \! u^{(1),\{1,2\}}_{\{2,3\}} \! = \! u^{(1),\{1,2\}}_{\{1\}} \! + \! u^{(1),\{1,2\}}_{\{1,3\}}, \\
v_{\{1,3\}}\! &= \! u^{(1),\{1,3\}}_{\{3\}} \! + \! u^{(1),\{1,3\}}_{\{2,3\}} \! = \! u^{(1),\{1,3\}}_{\{1\}} \! + \! u^{(1),\{1,3\}}_{\{1,2\}}, \\
v_{\{2,3\}} \! &= \! u^{(1),\{2,3\}}_{\{3\}} \! + u^{(1),\{2,3\}}_{\{1,3\}} \! + u^{(2),\{2,3\}}_{\{3\}}, \\
&= u^{(1),\{2,3\}}_{\{2\}} \! + \! u^{(1),\{2,3\}}_{\{1,2\}} \! + \! u^{(2),\{2,3\}}_{\{2\}}, \\
\! \! \! v_{\{1,2,3\}} \! &= \! u^{(1),\{1,2,3\}}_{\{2,3\}} \! = \! u^{(1),\{1,2,3\}}_{\{1,3\}} \! = \! u^{(1),\{1,2,3\}}_{\{1,2\}},
\end{align}
\end{subequations} 
the delivery completion constraints for user $3$ are defined as
\begin{subequations}
\begin{align}
&v_{\{3\}}^{(1)} \! + \! \big( u^{(1),\{1,3\}}_{\{1\}} \! + \! u^{(1),\{1,3\}}_{\{1,2\}}\big) \! + \! \big( u^{(1),\{2,3\}}_{\{2\}} \! + \! u^{(1),\{2,3\}}_{\{1,2\}} \big) \! + \nonumber \\ &u^{(1),\{1,2,3\}}_{\{1,2\}} \! + \! \big( a_{\{3\}}^{(1)} \! + \! a_{\{1,3\}}^{(1)} \! + a_{\{2,3\}}^{(1)} \! + a_{\{1,2,3\}}^{(1)} \big) \! \geq \! r_1, \\
&v_{\{3\}}^{(2)}  +  u^{(2),\{2,3\}}_{\{2\}}  +  a_{\{3\}}^{(2)}  +  a_{\{2,3\}}^{(2)}  \geq  r_2 \! - r_1, \\
&v_{\{3\}}^{(3)} + a_{\{3\}}^{(3)} \geq r_3 \! - \! r_2,
\end{align}
\end{subequations}
and the redundancy constraints are defined as
\begin{subequations}
\begin{align}
u^{(1),\{1,3\}}_{\{1,2\}}  +  u^{(1),\{2,3\}}_{\{1,2\}}  +  u^{(1),\{1,2,3\}}_{\{1,2\}}  &\leq   a^{(1)}_{\{1,2\}}, \\
u^{(1),\{1,2\}}_{\{1,3\}}+u^{(1),\{2,3\}}_{\{1,3\}}+u^{(1),\{1,2,3\}}_{\{1,3\}} &\leq a^{(1)}_{\{1,3\}}, \\
u^{(1),\{1,2\}}_{\{2,3\}}+u^{(1),\{1,3\}}_{\{2,3\}}+u^{(1),\{1,2,3\}}_{\{2,3\}} &\leq a^{(1)}_{\{2,3\}}. 
\end{align}
\end{subequations}
We denote the vector representation of the transmission variables $ \{ v_{\mc T} \}_{ \mc T}, \{v_{\{k\}}^{(l)}\}_{k,l} $ and the assignment variables $\{ u^{(l),\mc T}_{\mc S} \}_{l,\mc T, \mc S}$ by $\bm v$ and $\bm u $, respectively.
\section{Formulation}
\vspace{-.05in}
In this section, we demonstrate that the problem of minimizing the worst-case delivery load by optimizing over the users' cache sizes, uncoded placement, and linear delivery, can be formulated as a linear program. In particular, given the target rate vector $\bm r$, the total memory budget $m_{\text{tot}}$, and $N \geq K$, the minimum worst-case delivery load under uncoded placement and linear delivery, $R^*_{\mathfrak{A},\mathfrak{D}}(m_{\text{tot}},\bm r) $, is characterized by 
\begin{subequations} \label{eqn_opt_budget}
\begin{align}
\textit{\textbf{O1}:}  \qquad   &  \min_{\bm a^{(l)},\bm u ,\bm v,\bm m^{(l)}}  
& & \sum_{\mc T \subsetneq_{\phi} [K] } \! \! \! \! v_{\mc T} \\
& \text{subject to}
& & \bm a^{(l)} \in \mathfrak{A}^{(l)}(\bm m^{(l)},\bm r), \forall l \in [K]\\
& & & \! \! \! (\bm u, \bm v) \in \mathfrak D(\bm a^{(1)},\dots,\bm a^{(K)},\bm r),\\
& & & \sum_{k=1}^{K}\sum_{l=1}^{k} m_k^{(l)} = m_{\text{tot}}, 
\end{align}
\end{subequations}
where $\mathfrak{A}^{(l)}(\bm m^{(l)},\bm r) $ is the set of uncoded placement schemes in layer $l$ defined in (\ref{eqn_feas_alloc}) and $\mathfrak D(\bm a^{(1)},\dots,\bm a^{(K)},\bm r) $ is the set of feasible linear delivery schemes defined by (\ref{eqn_dlv_cnst1})-(\ref{eqn_dlv_cnst4}).

We can also solve the problem of designing the caching scheme and the memory allocation over the layers for systems with fixed cache sizes as in \cite{yang2016coded}. In particular, for fixed cache sizes $\bm m$, the minimum worst-case delivery load under uncoded placement and linear delivery, $R^*_{\mathfrak{A},\mathfrak{D}}(\bm m,\bm r) $, is characterized by
\begin{subequations} \label{eqn_opt_fixed_Mem}
\begin{align}
\textit{\textbf{O2}:}  \qquad   &  \min_{\bm a^{(l)},\bm u ,\bm v,\bm m^{(l)}}  
& & \sum_{\mc T \subsetneq_{\phi} [K] } \! \! \! \! v_{\mc T} \\
& \text{subject to}
& & \bm a^{(l)} \in \mathfrak{A}^{(l)}(\bm m^{(l)},\bm r), \forall l \in [K]\\
& & & \! \! \! (\bm u, \bm v) \in \mathfrak D(\bm a^{(1)},\dots,\bm a^{(K)},\bm r),\\
& & & \! \! \sum_{l=1}^{k} m_k^{(l)} = m_k, \forall k \in [K]. 
\end{align}
\end{subequations}
In contrast to the formulation in reference \cite{yang2016coded}, observe that in (\ref{eqn_opt_fixed_Mem}), we jointly design the caching scheme and the memory allocation over the layers, and exploit multicast opportunities over different layers. 
\section{Optimal Cache Allocation}
In this section, we first characterize the solution to the optimization problem in (\ref{eqn_opt_budget}), i.e., we find the achievable worst-case delivery load assuming uncoded placement and linear delivery. Then, we present a lower bound on the trade-off between the minimum worst-case delivery load under any caching scheme and the cache memory budget.

\begin{theorem}\label{thm_mtot}
Given the target rate vector $\bm r$, $N \geq K$, and the total memory budget $m_{\text{tot}}= \sum_{l=1}^{K} t_l f_l$, where $f_1 \triangleq r_1 $, $f_l \triangleq r_l-r_{l-1}$ for $l>1$, $t_l \in \{0,1,\dots,K-l+1\}$ and $ t_{l+1} \leq t_l \leq t_{l-1} $, the minimum achievable worst-case delivery load under uncoded placement is given by
\begin{align}\label{eqn_thm_mtot}
R^*_{\mathfrak{A},\mathfrak{D}}(m_{\text{tot}},\bm r)\!= \!R^*_{\mathfrak{A}}(\bm r,m_{\text{tot}}) \! = \! \sum_{l=1}^{K} \dfrac{(K-l+1)-t_l}{1+t_l}f_l. 
\end{align}
Furthermore, for general $m_{\text{tot}} \in [0,\sum_{k=1}^{K}r_k]$, $R^*_{\mathfrak{A},\mathfrak{D}}(\bm r,m_{\text{tot}})$ is defined by the lower convex envelope of these points.
\end{theorem}
\begin{proof} In Appendix \ref{app_thm_mtot}, we show that the optimal solution to the optimization problem in (\ref{eqn_opt_budget}) achieves the lower convex envelope of the delivery load points in (\ref{eqn_thm_mtot}). In particular, for $m_{\text{tot}}= \sum_{i=1}^{K} t_i f_i$ where $t_l \in \{0,1,\dots,K-l+1\}$ and $ t_{l+1} \leq t_l \leq t_{l-1} $, the optimal memory allocation is defined as $m_k^{(l)}=t_l f_l/(K-l+1) $ for $k \in \{l,\dots,K\}$, i.e., the users are assigned equal cache sizes in each layer. In turn, we apply the MaddahAli-Niesen caching scheme in each layer, i.e., in layer $l$ the placement phase is defined by
\begin{align}
|\tilde W^{(l)}_{j,\mc S}|&= \begin{cases} f_l/\binom{K-l+1}{t_l} F, \text{ for } |\mc S| =t_l,\\
0, \text{ otherwise}.
\end{cases}
\end{align}
and the multicast signals are defined by $\oplus_{j \in \mc T} \tilde W^{(l)}_{d_j,\mc T \setminus \{j\}}$ for $\mc T \subset \{l,\dots,K\}$ and $|\mc T| =t_l+1$.

In general, any $m_{\text{tot}} \in [0, \sum_{i=1}^{K}r_i] $ can be represented as $ \sum_{i=1}^{K} t_i f_i$ where $t_i=x,$ for $ i=[y]$, $t_y=x-1+\alpha$, $t_i=x-1, $ for $  i=\{y+1,\dots,K-x+1\} $, $t_i=K-i+1, $ for $ i=\{K-x+2,\dots,K\} $, for some $x\in [K]$, $y \in [K-x+1]$ and $0<\alpha<1$. In particular, we have 
\begin{align}\label{eqn_def_thr_mtot1}
\! \! x \triangleq \begin{cases}1, \ 0 < m_{\text{tot}} \leq \sum\limits_{i=1}^{K}f_i, \\
\vdots \\
j, (j\!\!-\!\!1)\!\! \!\! \sum\limits_{i=1}^{K-j+1}\!\! \! \! f_i \!+\!\! \!\! \!\! \!\!\sum\limits_{i=K-j+2}^{K}\!\!\!\!\!\!\!(K\!\!-\!\!i\!\!+\!\!1)f_i \! <\! m_{\text{tot}} \! \leq j\!\! \sum\limits_{i=1}^{K-j}\!\! f_i \\ \vdots \ \ \ \ \ \ \ \ \ \ \ \ \ \ \ \ \ \  \ \ \ \ \ \ \ \  \ \    +\!\! \!\! \!\! \sum\limits_{i=K-j+1}^{K}\!\!\!\!\!(K\!\!-\!\!i\!\!+\!\!1)f_i, \\
K, \ \! (K\!\!-\!\!1)f_1 \!+\!\! \sum\limits_{i=2}^{K}(K\!\!-\!\!i\!\!+\!\!1)f_i \! <\! m_{\text{tot}} \! \leq \! \sum\limits_{i=1}^{K} \! r_i.
\end{cases}\!\!\!\!
\end{align}
and for a given $x$, we have $y=b$ if 
\begin{align}\label{eqn_def_thr_mtot2}
x\sum_{i=1}^{b-1}f_i&+(x\!-\!1)\!\!\!\!\sum\limits_{i=b}^{K-x+1}\!\!\!\!f_i\!+\!\! \!\!\!\!\!\sum\limits_{i=K-x+2}^{K}\!\!\!\!\!\!(K\!\!-\!\!i\!\!+\!\!1)f_i \!< \!m_{\text{tot}} \!\leq \!  \nonumber \\ &x\sum_{i=1}^{b}f_i+(x\!-\!1)\!\!\!\!\sum\limits_{i=b+1}^{K-x+1}\!\!\!\!f_i\!+\!\! \!\!\!\!\!\sum\limits_{i=K-x+2}^{K}\!\!\!\!\!\!(K\!\!-\!\!i\!\!+\!\!1)f_i.
\end{align}
In turn, any $m_{\text{tot}} \in [0,\sum_{k=1}^{K}r_k]$ can be represented as 
\begin{align}\label{eqn_def_thr_mtot3}
m_{\text{tot}}  &=  x \sum_{l=1}^{y-1}  f_l+(x- 1+\alpha)f_y +(x-1) \sum_{l=y+1}^{K-x+1}f_l + \nonumber \\ &  \ \ \  \sum_{l=K-x+2}^{K}  (K-l+1)f_l,
\end{align}
and the corresponding minimum worst-case delivery load under uncoded placement and linear delivery is given by
\begin{align}
\! \! &R^*_{\mathfrak{A},\mathfrak{D}}(m_{\text{tot}},\bm r) \! = \! \sum_{l=1}^{y-1} \dfrac{K\!-\!l\!-\!x\!+\!1}{x+1}f_l \!+ \! \! \! \sum_{l=y+1}^{K-x+1} \dfrac{K\!-\!y\!-\!x\!+\!2}{x}f_l+ \nonumber \\
&\qquad \quad \bigg(\dfrac{2(K\!-\!y)\!-\!x\!+\!3}{x+1}\!-\!\dfrac{(K\!-\!y\!+\!2)(x\!-\!1\!+\!\alpha)}{x(x\!+\!1)}\bigg)f_y. 
\end{align}
\end{proof}
For example, for $K=3$, the corner points of the delivery load cache budget trade-off under uncoded placement in (\ref{eqn_thm_mtot}) are defined as follows
%
\begin{align}\label{eqn_R_achv_K3}
\!\! \!\! R^*_{\mathfrak{A},\mathfrak{D}}(m_{\text{tot}},\bm r) \! = \! \begin{cases}r_1\!+\!r_2\!+\!r_3, \text{for } m_{\text{tot}}=0, \\
r_2\!+\!r_3\!-\!r_1, \text{for } m_{\text{tot}} \! = \!r_1, \\
r_1/2\!+\!r_3\!-\!r_2/2, \text{for } m_{\text{tot}} \! = \!r_2, \\
r_1/2\!+\!r_2/2, \text{for } m_{\text{tot}}\! =\! r_3, \\
r_2/2\!-\!r_1/6, \text{for } m_{\text{tot}}\! =\! r_1\!+\!r_3, \\
r_1/3, \text{for } m_{\text{tot}}\! =\! r_2\!+\!r_3, \\
0, \text{for } m_{\text{tot}}\! =\! r_1\!+\!r_2\!+\!r_3,
\end{cases}\!\!\!\!
\end{align}
which are illustrated in Fig. \ref{fig:R_vs_mtot_K3}, for $\bm r=[0.5, \ \! 0.7, \ \! 1]$.
%



Next, we extend the lower bound on the delivery load for systems with fixed cache sizes in \cite[Theorem 2]{yang2016coded} to systems with cache memory budget.

\begin{proposition}\label{prop_cutset_mtot} (Extension of \cite[Theorem 2]{yang2016coded})
Given $N \geq K$, the target rate vector $\bm r$, and the total memory budget $m_{\text{tot}} \in [0,\sum_{k=1}^{K}r_k]$, the infimum over all achievable delivery loads $R^*(m_{\text{tot}},\bm r)$ is lower bounded by
\begin{align}\label{eqn_cutset_mtot}
\! \! \min\limits_{\bm m \in \mc M(m_{\text{tot}}, \bm r)} \! \left\lbrace \! \max\limits_{\mc U \subset [K]} \! \left\lbrace \! \sum_{k \in \mc U} r_k -\dfrac{N \sum\limits_{k \in \mc U}m_k}{\left\lfloor N/|\mc U|\right\rfloor}\right\rbrace\! \! \right\rbrace,
\end{align}
where $\mc M(m_{\text{tot}}, \bm r)=\Big\{\bm m \big\vert \ \! 0 \! \leq \! m_k \! \leq \! r_k, \ \! \sum\limits_{k=1}^{K} m_k = m_{\text{tot}} \Big\} $. 
\end{proposition}
For the three-user case, the cut-set bound can be simplified as follows. 
\begin{corollary} Given $K=3$, $N \geq 3$, the target rate vector $\bm r$, and the total memory budget $m_{\text{tot}} \in [0,\sum_{k=1}^{K}r_k]$, the infimum over all achievable delivery loads $R^*(m_{\text{tot}},\bm r)$ is lower bounded by
\begin{align}\label{eqn_cutset_mtot_3Ue}
&\max \bigg\{ \sum_{l=1}^{3} r_l - \dfrac{N}{\lfloor N/3\rfloor} m_{\text{tot}}, \ \! \dfrac{\lfloor N/2\rfloor (r_1 +r_2)+ N r_3 }{N+ \lfloor N/2\rfloor}- \nonumber \\ &\qquad \qquad \qquad \dfrac{N}{N+\lfloor N/2\rfloor} m_{\text{tot}}, \ \! \dfrac{1}{3} \big( \sum_{l=1}^{3} r_l - m_{\text{tot}}\big) \bigg\}.
\end{align}
\end{corollary}



\section{Discussion and Numerical Results}
In this section, we first consider systems with total cache memory budget and present numerical examples that illustrate the optimal solution of (\ref{eqn_opt_budget}). Then, we consider systems with fixed cache sizes and illustrate the gain achieved by our caching scheme compared to the schemes proposed in \cite{yang2016coded}.

\subsection{Systems with Total Cache Memory Budget}

In Fig. \ref{fig:R_vs_mtot_K3}, we compare the achievable worst-case delivery load $R^*_{\mathfrak{A}}(m_{\text{tot}}, \bm r)$ defined in Theorem \ref{thm_mtot} with the lower bound on $R^*(m_{\text{tot}}, \bm r)$ defined in Proposition \ref{prop_cutset_mtot}, for $K=N=3$ and the target rates $\bm r =[0.5, \ \! 0.7, \ \! 1]$. In particular, for $K=3$ the corner points of the achievable delivery load are defined in (\ref{eqn_R_achv_K3}) and the cut-set bound in (\ref{eqn_cutset_mtot_3Ue}) for $N=3$ is 
\begin{align}
\! \! \! \! \! \!   \max \! \! \left\{ \! \sum_{l=1}^{3} \! \! r_l \!- \!3 m_{\text{tot}}, \dfrac{r_1 \! + \! r_2 \! + \! 3 r_3 \!- \! 3 m_{\text{tot}}}{4}, \frac{1}{3} \big(\! \sum_{l=1}^{3} \! \! r_l \! - \! m_{\text{tot}} \big) \! \! \right\} \! \! .\! \!
\end{align}
\begin{remark}
We observe that for large memory budget, $ \sum_{l=2}^{K} r_l \leq m_{\text{tot}} \leq \sum_{l=1}^{K} r_l$, the cut-set bound can be achieved and the delivery load $R^*(m_{\text{tot}}, \bm r)= \frac{1}{K} \big( \sum_{l=1}^{K}  \ r_l \! - \! m_{\text{tot}} \big)$.
\end{remark}
\begin{figure}[t]
\includegraphics[scale=0.5]{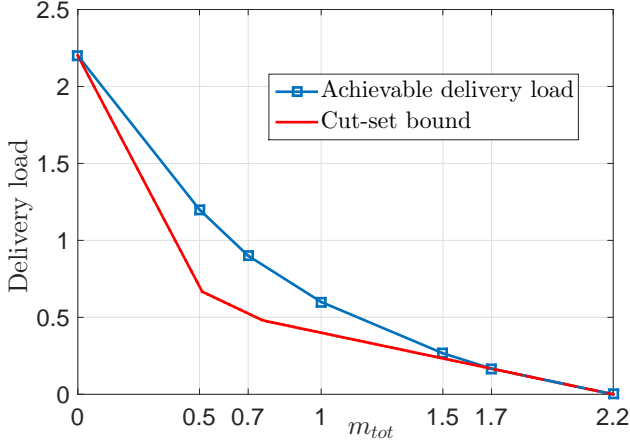}
\centering
\caption{Comparing the delivery load $R^*_{\mathfrak{A}}(m_{\text{tot}}, \bm r)$ in Theorem \ref{thm_mtot} with the cut-set bound in Proposition \ref{prop_cutset_mtot}, for $K=N=3$, and $\bm r =[0.5, \ \! 0.7, \ \! 1]$.}\label{fig:R_vs_mtot_K3}
\end{figure}

Recall the notation $f_l=(r_l-r_{l-1})$ and $r_0=0$. Fig. \ref{fig:m_alloc_K3} shows the optimal allocation for the users' cache sizes that correspond to the achievable delivery load in Fig. \ref{fig:R_vs_mtot_K3}. From Fig. \ref{fig:m_alloc_K3}, we observe the optimal memory allocation follows a threshold policy. In particular, we have 
\begin{itemize}
\item If $m_{\text{tot}} \! = \! \alpha f_1$ and $\alpha \! \in \! [0,1]$, then $\bm m \!= \![\alpha f_1/3, \ \! \alpha f_1/3, \ \! \alpha f_1/3]$.
\item If $m_{\text{tot}}=f_1 \! + \! \alpha f_2$ and $\alpha \in [0,1]$, then $\bm m = [f_1/3, \ \! f_1/3\!+\!\alpha f_2/2, \ \! f_1/3\!+\!\alpha f_2/2]$.
\item If $m_{\text{tot}}=f_1 \! + \! f_2 \! + \! \alpha f_3$ and $\alpha \in [0,1]$, then $\bm m = [f_1/3, \ \! f_1/3\!+\!f_2/2, \ \! f_1/3\!+\!f_2/2 \! + \! \alpha f_3]$.
\item If $m_{\text{tot}}=(1\!+\!\alpha)f_1 \! + \! f_2 \! + \! f_3 $ and $\alpha \in [0,1]$, then $\bm m = [(1\!+\!\alpha)f_1/3, \ \! (1\!+\!\alpha)f_1/3\!+\!f_2/2, \ \! (1\!+\!\alpha)f_1/3\!+\!f_2/2 \! + \! f_3]$.
\item If $m_{\text{tot}}=2f_1 \! + \! (1\!+\!\alpha)f_2 \! + \! f_3 $ and $\alpha \in [0,1]$, then $\bm m = [2f_1/3, \ \! 2f_1/3\!+\!(1\!+\!\alpha)f_2/2, \ \! 2f_1/3\!+\!(1\!+\!\alpha)f_2/2 \! + \! f_3]$.
\item If $m_{\text{tot}}=(2\!+\!\alpha)f_1 \! + \! 2f_2 \! + \! f_3 $ and $\alpha \in [0,1]$, then $\bm m = [(2\!+\!\alpha)f_1/3, \ \! (2\!+\!\alpha)f_1/3\!+\!f_2, \ \! (2\!+\!\alpha)f_1/3\!+\!f_2 \! + \! f_3]$.
\end{itemize}

\begin{figure}[t]
\includegraphics[scale=0.5]{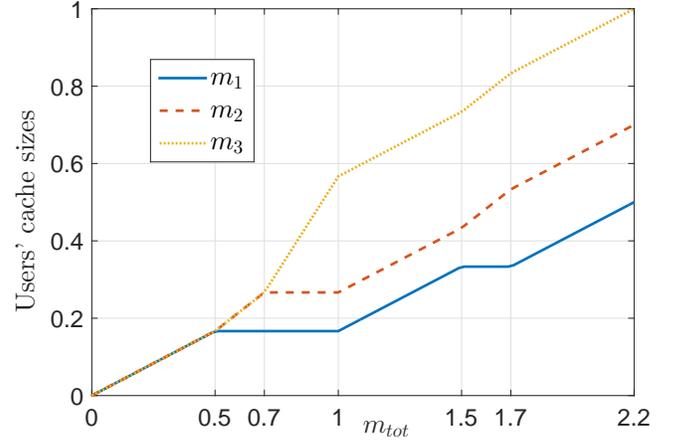}
\centering
\caption{The memory allocation corresponding to the delivery load in Theorem \ref{thm_mtot} for $\bm r =[0.5, \ \! 0.7, \ \! 1]$.}\label{fig:m_alloc_K3}
\end{figure}

\subsection{Systems with Fixed Cache Sizes}

For systems with fixed cache sizes, we compare the delivery load achieved by our scheme which performs joint design of the memory allocation over the layers and the caching scheme, with the two separation based schemes in \cite{yang2016coded}. In reference \cite{yang2016coded}, the memory allocation over the layers follows one of the following heuristic methods:
\begin{enumerate}
\item \textit{Proportional cache allocation} (PCA), where user $k$ cache size for layer $l$ is defined as $m_k^{(l)}=f_l m_k / r_k$. 
\item \textit{Ordered cache allocation} (OCA), where user $k$ allocates all of its cache to the first $b$ layers if $  r_{b-1}< m_k < r_b$. 
\end{enumerate}
In turn, in layer $l$, we have a caching system that consists of $K-l+1$ users with fixed cache sizes $\bm m^{(l)} $. Reference \cite{yang2016coded} has proposed designing the caching scheme in each layer separately, i.e., the multicast signals utilize the side-information from one layer. In particular, in layer $l$, the cache sizes $\bm m^{(l)} $ are further divided into sublayers of equal size and the MaddahAli-Niesen scheme is applied on each sublayer. Additionally, reference \cite{yang2016coded} has formulated an optimization problem in order to identify the optimal distribution of $f_l$ over the sublayers. We refer to this scheme as the \textit{layered scheme}.


Fig. \ref{fig:comp_fix_mem} shows the delivery load obtained from (\ref{eqn_opt_fixed_Mem}), the delivery load achieved by PCA/OCA combined with the layered scheme, and the cut-set bound in \cite[Theorem 2]{yang2016coded}, for $K=N=3$, $\bm r =[0.5, \ \! 0.8, \ \! 1]$, and $ m_k=0.8 \ \! m_{k+1}$. We observe that the delivery load achieved by (\ref{eqn_opt_fixed_Mem}) is lower than the one achieved by the schemes in \cite{yang2016coded}. This is attributed to the fact that we not only jointly optimize the memory allocation over the layers and the caching scheme, but also exploit the multicast opportunities over all layers. However, the performance improvement comes at the expense of higher complexity, since the dimension of the optimization problem in (\ref{eqn_opt_fixed_Mem}) grows exponentially with the number of users.

The next example illustrates the suboptimality of exploiting the multicast opportunities in each layer separately.

\begin{example}
For $\bm m=[0.1, \ \! 0.2, \ \! 0.6]$ and $ \bm r =[0.2, \ \! 0.3, \ \! 0.8]$, the optimal solution of (\ref{eqn_opt_fixed_Mem}) is as follows \newline
\textbf{Placement phase:} 
\begin{itemize}
\item Layer $1$: $W^{(1)}_{j}$ is divided into subfiles $\tilde W^{(1)}_{j,\{3\}} $ and $\tilde W^{(1)}_{j,\{1,2\}}$, such that $a^{(1)}_{\{3\}}=a^{(1)}_{\{1,2\}}=0.1 $.
\item Layer $2$: $W^{(2)}_{j}$ is stored at user $2$, i.e., $a^{(2)}_{\{2\}}=0.1$.
\item Layer $3$: $W^{(3)}_{j}$ is stored at user $3$, i.e., $a^{(3)}_{\{3\}}=0.5$.
\end{itemize}
\textbf{Delivery phase:} We have two multicast transmissions $X_{\{1,3\}, \bm d}$ and $X_{\{2,3\}, \bm d}$, such that
\begin{itemize}
\item $X_{\{1,3\}, \bm d}= W^{(1),\{1,3\}}_{d_1,\{3\}} \oplus W^{(1),\{1,3\}}_{d_3,\{1,2\}}$, where $v_{\{1,3\}}=a^{(1)}_{\{3\}}=a^{(1)}_{\{1,2\}}=0.1$.
\item $X_{\{2,3\}, \bm d}= W^{(1),\{2,3\}}_{d_2,\{3\}} \oplus W^{(2),\{2,3\}}_{d_2,\{2\}}$, where $v_{\{2,3\}}=a^{(1)}_{\{3\}}=a^{(2)}_{\{2\}}=0.1$.
\end{itemize}
The optimal solution of (\ref{eqn_opt_fixed_Mem}) achieves the delivery load $R^*_{\mathfrak{A},\mathfrak{D}}(\bm m,\bm r) =0.2$ compared to $0.2167$ which is achieved by exploiting only the intra-layer multicast opportunities. 
\end{example}

\begin{figure}[t]
\includegraphics[scale=0.5]{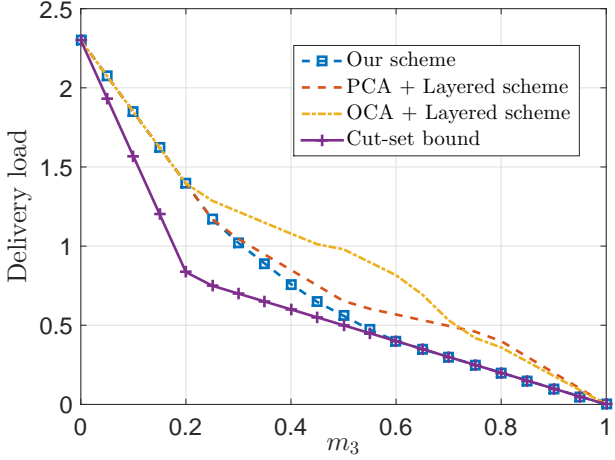}
\centering
\caption{Comparing the delivery load $R^*_{\mathfrak{A},\mathfrak{D}}(\bm m, \bm r)$ obtained from (\ref{eqn_opt_fixed_Mem}), the delivery load achieved by the two schemes in \cite{yang2016coded}, and the cut-set bound \cite[Theorem 2]{yang2016coded}, for $K=N=3$, $\bm r =[0.5, \ \! 0.8, \ \! 1]$, and $ m_k=0.8 \ \! m_{k+1}$.}\label{fig:comp_fix_mem}
\end{figure}

\section{Conclusion}\label{sec_con}
In this paper, we have studied coded caching systems with heterogeneous distortion requirements. In addition to designing the caching scheme, the server allocates the sizes of the cache memories at the end users, subject to a total cache budget constraint over the set of all users. Assuming uncoded placement and linear delivery policies, we have shown that the problem of minimizing the worst-case delivery load can be formulated as a linear program.

The optimal memory allocation has been shown to follow a threshold policy. Furthermore, we have observed that our solution matches the cut-set bound in the large total cache memory budget region. As a byproduct, we have proposed a novel caching scheme which, for fixed cache sizes, outperforms the state-of-art schemes \cite{yang2016coded} by exploiting the inter-layer multicast opportunities and jointly designing the cache contents and the partitioning of the caches over the layers.

%
%
%



\appendices
\section{Proof of Theorem \ref{thm_mtot} }\label{app_thm_mtot}
In order to prove that the optimal solution of (\ref{eqn_opt_budget}) achieves the delivery load in Theorem \ref{thm_mtot}, we first show the optimality of uniform cache allocation in each layer, i.e., $m_k^{(l)}=m^{(l)}, \forall k \in\{l,\dots,K\}$. Then, we show that the optimal memory allocation over the layers follows the threshold policy defined by (\ref{eqn_def_thr_mtot1})-(\ref{eqn_def_thr_mtot3}). 

Next lemma shows the optimality of allocating equal cache sizes in a caching system where the users request the same number of bits from the desired files \cite{maddah2014fundamental}, which is equivalent to equal distortion requirements at the users \cite{yang2016coded}.
\begin{lemma}\label{lemma_equal_mem}
For a caching system with $K$ users, $N \geq K$ files, and cache memory budget $m_{\text{tot}} \in [0,K]$, the minimum worst-case delivery load under uncoded placement
\begin{align}
R^*_{\mathfrak{A}}(m_{\text{tot}})= \max_{j \in [K]} \left\lbrace \dfrac{2K-j+1}{j+1}-\dfrac{(K+1)m_{\text{tot}}}{j(j+1)}\right\rbrace,
\end{align}
which is achieved by uniform memory allocation and the MaddahAli-Niesen caching scheme \cite{maddah2014fundamental}.
\end{lemma}
\begin{proof}
For any memory allocation $\bm m=[m_1,\dots,m_K]$, 
\begin{subequations}
\begin{align}
\! \! \! \! \! R^*_{\mathfrak{A}}(\bm m) \geq 	\!  \!    & \max_{\lambda_{0} \in \mathbb{R},\lambda_{k} \geq 0}  
	& & -\lambda_{0} - \sum_{k=1}^{K} m_k \lambda_{k} \\
	& \! \! \! \! \! \!  \text{subject to}
	& &  \! \! \! \! \! \! \! \! \! \! \! \lambda_0 \!  + \!  \sum_{k \in \mc S} \!  \lambda_k \!  + \!  \dfrac{K\! -\! |\mc S|}{|\mc S|\! +\! 1} \geq 0, \ \! \forall \mc S \!  \subset \! [K], 
\end{align}
\end{subequations}
which is obtained by considering the average cut in the lower bound under uncoded placement in \cite[Theorem 1]{ibrahim2018coded}. In turn, by considering $\lambda_k= \lambda, \ \! \forall k \in [K]$, we get
\begin{align}
R^*_{\mathfrak{A}}(m_{\text{tot}}) \! &\geq \max_{\lambda \geq 0} \min_{j \in \{0,\dots,K\}} \! \! \left\lbrace  \dfrac{K-j}{j+1} - \lambda(j  -  m_{\text{tot}}) \right\rbrace, \\
&=\max_{j \in [K]} \left\lbrace \dfrac{2K-j+1}{j+1}-\dfrac{(K+1)m_{\text{tot}}}{j(j+1)}\right\rbrace,
\end{align}
where $m_{\text{tot}}= \sum_{k=1}^{K} m_k$.
\end{proof}
Building on Lemma \ref{lemma_equal_mem}, given $\bm r$ and any memory allocation $\bm m^{(1)}, \dots,\bm m^{(K)}$, the minimum worst-case delivery load under uncoded placement, $R^*_{\mathfrak{A}}(\bm m^{(1)}, \dots, \bm m^{(K)},\bm r)$ satisfies
\begin{align}
&R^*_{\mathfrak{A}}( \bm m^{(1)}, \dots, \bm m^{(K)},\bm r) \geq \sum_{l=1}^{K} (r_l -r_{l-1}) \nonumber \\ &\max_{j \in [K-l+1]} \!\!\bigg\lbrace \! \dfrac{2(K\!-\!l\!+\!1)\!-\!j\!+\!1}{j+1} \! - \! \dfrac{(K\!-\!l\!+\!2) \! \sum\limits_{k=l}^{K} \! m^{(l)}_k}{j(j+1)} \!\bigg\rbrace .
\end{align}
Furthermore, this lower bound is achievable if we consider $m_k^{(l)}=m^{(l)}, \forall k \in\{l,\dots,K\}$ and apply the MaddahAli-Niesen caching scheme \cite{maddah2014fundamental} on each layer. In turn, the optimization problem in (\ref{eqn_opt_budget}) simplifies to the problem of allocating the memory over the layers, which is defined as
\begin{subequations} \label{eqn_opt_budget_simpl}
\begin{align}
&  \min_{ t_1,\dots,t_K}  
& & \sum_{l=1}^{K} \chi_l f_l\\
& \text{subject to}
& & \sum_{l=1}^{K} t_l f_l \leq m_{\text{tot}}, \\
& & & 0\leq t_l \leq K-l+1,
\end{align}
\end{subequations}
where $m^{(l)}= t_l f_l/(K-l+1)$ and
\begin{align}
\chi_l \!\! \triangleq \max_{j \in [K-l+1]} \left\lbrace \dfrac{2(K\!-\!l\!+\!1)\!-\!j\!+\!1}{j+1}-\dfrac{(K\!-\!l\!+\!2)t_l}{j(j+1)}\right\rbrace.
\end{align}
Finally, we can show that the optimal solution to (\ref{eqn_opt_budget_simpl}) achieves the delivery load in Theorem \ref{thm_mtot}, by solving the dual of the linear program in (\ref{eqn_opt_budget_simpl}).

\bibliographystyle{IEEEtran}
\bibliography{IEEEabrv,references}
\end{document}